\documentclass[11pt]{amsart}

\usepackage[left=3cm, right=3cm, top=3cm, bottom=3cm]{geometry}
\usepackage{amsmath,amsfonts,amssymb,amsthm}
\usepackage{mathrsfs}
\usepackage{url}
\usepackage{graphicx}
\usepackage{subfigure}

\renewcommand{\phi}{\varphi}
\newcommand{\iy}{\infty}

\DeclareMathOperator{\trace}{Tr}
 
\DeclareMathOperator{\I}{I}

\DeclareMathOperator{\Lat}{Lat}
\DeclareMathOperator{\cov}{cov}

\newcommand{\Haar}{\mathfrak h}

\renewcommand{\Re}{\mathop{\rm{Re}}\nolimits}
\renewcommand{\Im}{\mathop{\rm{Im}}\nolimits}
\renewcommand{\leq}{\leqslant}
\renewcommand{\geq}{\geqslant}

\newcommand{\norm}[1]{\left\Vert #1\right\Vert}
\newcommand{\abs}[1]{\left| #1\right|}

\newcommand{\R}{\mathbb{R}}
\newcommand{\C}{\mathbb{C}}
\newcommand{\E}{\mathbb{E}}
\renewcommand{\P}{\mathbb{P}}

\newcommand{\orth}{\bot}
\newcommand{\scalar}[2]{\langle #1 , #2\rangle}
\newcommand{\ketbra}[2]{| #1 \rangle \langle #2 |}

\newcommand{\B}{\mathcal{B}}
\renewcommand{\H}{\mathcal{H}}
\newcommand{\K}{\mathcal{K}}
\newcommand{\U}{\mathcal{U}}

\newcommand{\M}{\mathcal{M}}
\newcommand{\MD}{\mathcal{M}^{1, +}}
\newcommand{\Msa}{\mathcal{M}^\text{sa}}
\renewcommand{\S}{\mathcal{S}}

\newcommand{\T}{\mathbb{T}}

\theoremstyle{plain}
\newtheorem{thm}{Theorem}[section] 
\newtheorem{cor}[thm]{Corollary}
\newtheorem{prop}[thm]{Proposition}
\newtheorem{lem}[thm]{Lemma}

\theoremstyle{definition}
\newtheorem{defn}[thm]{Definition}

\theoremstyle{remark}
\newtheorem{rem}[thm]{Remark}

\newtheorem{eg}[thm]{Example}

\title{Random repeated quantum interactions\\and random invariant states}

\author[I. Nechita]{Ion Nechita}
\address{Universit\'e de Lyon, Institut Camille Jordan, 43 blvd du 11 novembre 1918, F-69622 Villeurbanne-Cedex, France} \email{nechita@math.univ-lyon1.fr}

\author[C. Pellegrini]{Cl\'ement Pellegrini}
\address{School of Physics and National Institue for Theoretical Physics, University of KwaZulu Natal, Private Bag X54001, Durban 4000, South Africa}
\email{pelleg@math.univ-lyon1.fr}


\begin{document}

\begin{abstract}
%
%
We consider a generalized model of repeated quantum interactions, where a system $\mathcal{H}$ is interacting in a random way with a sequence of independent quantum systems $\mathcal{K}_n, n \geq 1$. Two types of randomness are studied in detail. One is provided by considering Haar-distributed unitaries to describe each interaction between $\mathcal{H}$ and $\mathcal{K}_n$. The other involves random quantum states describing each copy $\mathcal{K}_n$. In the limit of a large number of interactions, we present convergence results for the asymptotic state of $\mathcal{H}$. This is achieved by studying spectral properties of (random) quantum channels which guarantee the existence of unique invariant states. Finally this allows to introduce a new physically motivated ensemble of random density matrices called the \emph{asymptotic induced ensemble}.
\end{abstract}

\maketitle

\section{Introduction}

Initially introduced in \cite{attal_pautrat} as a discrete approximation of Langevin dynamics, the model of repeated quantum interactions has found since many applications (quantum trajectories, stochastic control, etc.). In this work we generalize this model by allowing \emph{random} interactions at each time step. Our main focus is the long-time behavior of the reduced dynamics. 

Our viewpoint is that of Quantum Open Systems, where a ``small'' system is in interaction with an inaccessible environment (or an auxiliary system). We are interested in the reduced dynamics of the small system, which is described by the action of quantum channels. When repeating such interactions, under some mild conditions on the spectrum of the quantum channel, we show that the successive states of the small system converge to the invariant density matrix of the channel.

These considerations motivated us to consider random invariant states, and we introduce a new probability measure on the set of density matrices. There exists extensive literature \cite{braunstein, zyc_sommers, nechita, zyc_book} on what is a ``typical'' density matrix. There are two general categories of such probability measures on $\MD_d(\C)$: measures that come from metrics with statistical significance and the so-called ``induced measures'', where density matrices are obtained as partial traces of larger, random pure states. Our construction from Section \ref{sec:fixed} falls into the second category, since our model involves an open system in interaction with a chain of ``auxiliary'' systems. 

Next, we introduce two models of random quantum channels. In the first model, we allow for the states of the auxiliary system to be random. In the second one, the unitary matrices acting on the coupled system are assumed random, distributed along the Haar invariant probability on the unitary group, and independent between different interactions. Since the (random) state of the system fluctuates, almost sure convergence does not hold, and we state results in the ergodic sense. 

The article is structured as follows. The Section \ref{sec:rep_int} is devoted to presenting the model of quantum repeated interactions and its description via quantum channels. Section \ref{sec:spectral} contains some general facts about the spectra of completely positive maps, as well as some related tools from matrix analysis. Next, in Section \ref{sec:fixed} we study our first model, where the interaction unitary is a fixed, deterministic matrix. We prove that, under some assumptions on the spectrum of the quantum channel, the state of the system converges to the invariant state of the channel. It is at this time that we introduce the new ensemble of random density matrices, by transporting the unitary Haar measure via the application which maps a channel to its invariant state. The final two sections are devoted to introducing two models of random quantum channels, one where the interaction unitary is constant and the auxiliary states are i.i.d. density matrices (Sec. \ref{sec:random_env}) and another where the interaction unitaries are independent and Haar distributed (Sec. \ref{sec:iid_unitaries}). 

We introduce now some notation and recall some basic facts and terminology from quantum information theory. We write $\Msa_d(\C)$ for the set of self-adjoint $d \times d$ complex matrices and $\MD_d(\C)$ for the set of \emph{density matrices} (or states), $\MD_d(\C) = \{ \rho \in \Msa_d(\C) \, | \, \rho \geq 0, \trace[\rho] = 1\}$. Since our main focus is quantum information, all Hilbert spaces in this article are complex and finite dimensional. Scalar products are assumed linear in the second coordinate and, for two vectors $x \in \H, y \in \K$ we denote by $\ketbra x y \in \B(\K, \H)$ the map 
\[ \ketbra x y (z) = \scalar y z \cdot x, \quad \forall z \in \K.\]
An unit vector $x \in \H \simeq \C^d$ is called a \emph{pure state} and it is assimilated often with the orthogonal projection on $\C x$, $ \ketbra x x$. Finally, for a matrix $A \in \B(\H \otimes \K)  \simeq \B(\H) \otimes \B(\K)$, we define its \emph{partial trace} with respect to $\K$ as the unique element $B = \trace_\K[A] \in \B(\H)$ which verifies
\[ \trace[BX] = \trace [A (X \otimes \I_\K)], \quad \forall X \in \B(\H).\]

We shall also extensively use the \emph{Haar} (or uniform) measure $\Haar_d$ on the unitary group $\U(d)$; it is the unique probability measure which is invariant by left and right multiplication by unitary elements:
\[\forall V, W \in \U(d), \quad \forall f:\U(d) \to \C \text{ Borel}, \quad \int_{\U(d)}f(U) d\Haar_d(U) = \int_{\U(d)}f(VUW) d\Haar_d(U).\]

\section{The repeated quantum interaction model}\label{sec:rep_int}

In this introductory section we give a description of the physical model we shall use in the rest of the paper: \emph{repeated quantum interactions}. The setting, a system interacting repeatedly with ``independent'' copies of an environment, was introduced by S. Attal and Y. Pautrat in \cite{attal_pautrat} where it was shown that in the continuous limit (when the time between interactions approaches zero), the dynamics is governed by a quantum stochastic differential equation. A different model, where after each interaction an indirect quantum measurement of the system is performed, was considered by the second named author in \cite{pelleg_diffusive, pelleg_poisson} and shown to converge in the limit to the so-called stochastic Schr\"odinger equations. Here, we are concerned only with the discrete setting and with the limit of a large number of interactions. The study of random quantum trajectories is postponed to a later paper.

Consider a quantum system $\S$ described by a complex Hilbert space state $\H$. In realistic physical models, $\S$ is usually a quantum system with relatively few degrees of freedom and it represents the object of interest of our study; we shall refer to it as the \emph{small system}. Consider also another quantum system $\mathcal E$ which interacts with the initial small system $\S$. We shall call $\mathcal E$ the \emph{environment} and we denote by $\K$ its Hilbert state space. In this work we consider finite dimensional spaces $\H \simeq \C^d$ and $\K \simeq \C^{d'}$.

We shall eventually be interested in \emph{repeated} interactions between $\S$ and independent copies of $\mathcal E$, but let us start with the easier task of describing a single interaction between the ``small'' system and the environment. Assume that the initial state of the system is a product state $\sigma = \rho \otimes \beta$, where $\rho$ and $\beta$ are the respective states of the small system and the environment. The coupled system undergoes an unitary evolution $U$ and $U(\rho \otimes \beta)U^*$ is the global state after the interaction. The unitary operator $U$ comes from a Hamiltonian 
\[H_{tot}=H_\S \otimes \I+\I \otimes H_\mathcal E+H_{int},\]
where the operators $H_\S$ and $H_\mathcal E$ are the free Hamiltonians of the systems $\S$ and $\mathcal E$ respectively and $H_{int}$ represents the interaction Hamiltonian. We shall be interested in the situation where $H_{int} \neq 0$, otherwise there is no coupling and the system and the environment undergo separate dynamics. In this general case, the evolution unitary operator $U$ is given by
\[U = e^{-i\tau H_{tot}},\]
where $\tau > 0$ is the interaction time. Hence, the state of the coupled system $\S + \mathcal E$ after one interaction is given by
\[\sigma' = U(\rho \otimes \beta)U^*.\]
Since one is interested only in the dynamics of the ``small'' system $\S$, after taking the partial trace we obtain the final state of $\S$,
\begin{equation}\label{eq:single_interaction}
\rho' = \trace_\K[U(\rho \otimes \beta)U^*].
\end{equation}

We now move on to describe successive interactions between $\S$ and a chain of independent copies of $\mathcal E$. In order to do this, consider the countable tensor product
\[\K_{tot}=\bigotimes_{n=1}^\iy\mathcal{K}_n,\]
where $\K_n$ is the $n$-th copy of the environment ($\K_n \simeq \K \simeq \C^{d'}$). This setting can be interpreted in two different ways: globally, as an evolution on infinite dimensional countable tensor product $\H \otimes \K_{tot}$, or by discarding the environment, as a discrete evolution on $\B(\H) = \M_d(\C)$. Since we are interested only in the evolution of the ``small'' system, the latter approach is the better choice. From Eq.~(\ref{eq:single_interaction}), we obtain the recurrence relation
\begin{equation}\label{eq:rec_relation}
\rho_n = \trace_\K[U_n(\rho_{n-1} \otimes \beta_n)U_n^*],
\end{equation}
where $\rho_{n-1}, \rho_n \in \MD_d(\C)$ are the successive states of the system $\S$ at times $n-1$ and $n$, and $U_n$ and $\beta_n$ are the interaction unitary and respectively the state of the auxiliary system $\mathcal E$ for the $n$-th interaction. Note that at this stage we work in a general setting, without making any assumptions on the sequences $(U_n)_n$ and $(\beta_n)_n$.

We introduce now a more parsimonious description of repeated quantum interactions, via quantum channels. Recall that a linear map $\Phi:\M_d(\C) \to \M_d(\C)$ is called $k$-positive if the extended map $\Phi \otimes \I_k : \M_d(C) \otimes \M_k(\C) \to \M_d(C) \otimes \M_k(\C)$ is positive. $\Phi$ is called \emph{completely positive} if it is $k$-positive for all $k \geq 1$ (in fact $k=d$ suffices) and \emph{trace preserving} if $\trace[\Phi(X)] = \trace[X]$ for all $X \in \M_d(\C)$. By definition, a \emph{quantum channel} is a trace-preserving, completely positive linear map. The next proposition gives two very important characterizations of quantum channels. 

\begin{prop}[Stinespring-Kraus]\label{prop:stinespring_kraus}
A linear map $\Phi : \M_d(\C) \to \M_d(\C)$ is a quantum channel if and only if one of the following two equivalent conditions holds.
\begin{enumerate}
\item \textbf{(Stinespring dilation)} There exists a finite dimensional Hilbert space $\K = \C^{d'}$, a density matrix $\beta \in \MD_{d'}(\C)$ and an unitary operation $U \in \U(dd')$ such that
\[\Phi(X) = \trace_\K\left[ U(X \otimes \beta) U^* \right], \quad \forall X \in \M_d(\C).\]
\item \textbf{(Kraus decomposition)} There exists an integer $k$ and matrices $L_1, \ldots, L_k \in \M_d(\C)$ such that
\begin{equation}\label{eq:Kraus_decomposition}
\Phi(X) = \sum_{i=1}^{k} L_i X L_i^* ,\quad \forall X \in \M_d(\C)
\end{equation}
and
\[\sum_{i=1}^{k} L_i^* L_i = \I_d.\]
\end{enumerate}
\end{prop}
\begin{rem}
It can be shown that the dimension of the ancilla space in the Stinespring dilation theorem can be chosen $d'_0 = d^2$ and $\beta$ can be chosen to be a rank one projector. A similar result holds for the number of Kraus operators: one can always find a decomposition with $k=d^2$ operators. The \emph{Choi rank} of a quantum channel $\Phi$ is the least positive integer $k$ such that $\Phi$ admits a Kraus decomposition (\ref{eq:Kraus_decomposition}) with $k$ operators $L_i$. 
\end{rem}

We see now that Eq.~(\ref{eq:rec_relation}) can be re-written as
\[\rho_n = \Phi^{U_n, \beta_n}(\rho_{n-1}),\]
where $\Phi^{U, \beta}$ is the quantum channel
\begin{align*}
\MD_d(\C) &\to \MD_d(\C)\\
\rho &\mapsto \trace_\K[U(\rho \otimes \beta)U^*].
\end{align*}

After $n$ such interactions, the state of the system becomes
\begin{equation}\label{eq:compose_qc}
\rho_n = \Phi^{U_n, \beta_n} \circ \Phi^{U_{n-1}, \beta_{n-1}} \circ \cdots \circ \Phi^{U_1, \beta_1} \rho.
\end{equation}

Let us now consider a fixed channel $\Phi=\Phi^{U, \beta}$ and show that the Stinesping and Kraus form of $\Phi$ are connected in a simple fashion. To this end, start with the Stinespring form of $\Phi$ and pick some orthonormal bases $\{e_i\}_{i=1}^d$ and $\{f_j\}_{j=1}^{d'}$ of respectively $\H = \C^d$ and $\K = \C^{d'}$ such that the state of the environment $\beta$ diagonalizes:
\[\beta = \sum_{j=1}^{d'} b_j \ketbra{f_j}{f_j}.\]
Next, endow the product space $\H \otimes \K = \C^{dd'}$ with the basis 
\begin{equation}\label{eq:product_basis}
\{ {e_1} \otimes {f_1}, {e_2}\otimes {f_1}, \ldots,  {e_n} \otimes {f_1}, {e_1} \otimes {f_2}, \ldots, {e_n} \otimes {f_2}, \ldots, {e_n} \otimes {f_k}\}.
\end{equation}
This particular ordering of the tensor product basis was preferred in order to have a simple expression for the partial trace operation with respect to the environment $\K$. Indeed, if a matrix $A \in \M_{dd'}(\C)$ is written in this basis and viewed as a $d' \times d'$ matrix of blocks $A_{ij} \in \M_d(\C)$:
\[ A = \begin{pmatrix}
A_{11} & A_{12} & \cdots & A_{1d'}\\
A_{21} & A_{22} & \cdots & A_{2d'}\\
\vdots & \vdots & \ddots & \vdots\\
A_{d'1} & A_{d'2} & \cdots & A_{d'd'}
\end{pmatrix},\]
then the computation of the partial trace with respect to $\K=\C^{d'}$ reads
\[\trace_\K[A] = \trace_K\left[\sum_{i,j=1}^{d'}A_{ij} \otimes \ketbra{f_i}{f_j}\right] = \sum_{i,j=1}^{d'}A_{ij} \cdot \scalar{f_j}{f_i} = A_{11} + A_{22} + \cdots +A_{d'd'}.\]
In other words, the partial trace of $A$ over the environment $\K$ is simply the trace of the block-matrix, that is the sum of the diagonal blocks of $A$. We apply now these ideas to the Stinespring form of a quantum channel, $\Phi(X) =\trace_\K[U(X \otimes \beta)U^*]$. Written as a block matrix in the basis defined in Eq.~(\ref{eq:product_basis}), the matrix $X \otimes \beta$ is diagonal, with diagonal blocks given by $b_j X \in \M_d(\C)$. Writing $U \in \U(dd')$ in the same fashion and taking the partial trace, we obtain
\begin{equation}\label{eq:Kraus_from_U}
\Phi(X) = \trace_\K[U(X \otimes \beta)U^*] = \sum_{i,j=1}^{d'}b_j U_{ij} X U_{ij}^* = \sum_{i,j=1}^{d'}(\sqrt{b_j} U_{ij}) X (\sqrt{b_j} U_{ij})^*,
\end{equation}
where $U_{ij} \in \M_d(\C)$ are the blocks of the interaction unitary $U$. One recognizes a Kraus decomposition for $\Phi$, where the Kraus elements are rescaled versions of the blocks of the Stinespring matrix $U$. Moreover, if $\beta$ is a rank one projector then all the $b_j$'s are zero except one, hence the Kraus decomposition we obtained has $d'$ elements.

%

\section{Spectral properties of quantum channels}\label{sec:spectral}

Since we shall be interested in repeated applications of quantum channels, it is natural that spectral properties of these maps should play an important role in what follows.
One should note that most results of this section can be generalized to infinite dimensional Hilbert spaces.

The next lemma gathers some basic facts about quantum channels. Since quantum channels preserve the compact convex set of density matrices $\MD_d(\C)$, the first affirmation follows from the fixed point theorem of Markov-Kakutani \cite{ds}. The second and the third assertions are trivial (see \cite{ruskai} for further results on $L^p$ norms of quantum channels), and the last one is a consequence of 2-positivity.

\begin{lem}\label{lem:channel}
Let $\Phi : \M_d(\C) \to \M_d(\C)$ a quantum channel. Then
\begin{enumerate}
\item $\Phi$ has at least one invariant element, which is a density matrix;
\item $\Phi$ has trace operator norm of 1;
\item $\Phi$ has spectral radius of 1;
\item $\Phi$ satisfies the Schwarz inequality 
\[\forall X \in \M_d(\C), \quad \Phi(X)^* \Phi(X) \leq \norm{\Phi(1)} \Phi(X^*X).\]
\end{enumerate}
\end{lem}

If one looks at a channel $\Phi$ as an operator in the Hilbert space $\M_d(\C)$ endowed with the Hilbert-Schmidt scalar product, then one can introduce $\Psi$, the \emph{dual map} of $\Phi$. It is defined by the relation 
\[\trace[X \Phi(Y)] = \trace[\Psi(X) Y], \quad \forall X, Y \in \M_d(\C).\]
From Kraus decomposition $\Phi(X) = \sum L_i X L_i^*$, one can obtain a Kraus decomposition for the dual channel, $\Psi(X) = \sum L_i^* X L_i$. Note that the trace preserving condition for $\Phi$, $\sum L_i^* L_i = \I$ reads now $\Psi(\I) = \I$. Hence, the dual of a quantum channel is a unital (not necessarily trace-preserving) completely positive linear map. Using this idea, one can see that the partial trace operation $\trace_\K:\M_{dd'}(\C) \to \M_d(\C)$ is the dual of the tensoring operation $S_\K:\M_d(\C) \to \M_{dd'}(\C)$, $S(X)= X \otimes \I_{d'}$.

We now introduce some particular classes of positive maps which are known to have interesting spectral properties.

\begin{defn}
Let $\Phi : \M_d(\C) \to \M_d(\C)$ be a positive linear map. $\Phi$ is called \emph{strictly positive} (or positivity improving) if $\Phi(X) >0$ for all $X \geq 0$. $\Phi$ is called \emph{irreducible} if there is no projector $P$ such that $\Phi(P) \leq \lambda P$ for some $\lambda >0$.
\end{defn}

\begin{eg}\label{eg:unitary_conj}
Let $U \in \U(d)$ be a fixed unitary and consider the channel $\Phi: \M_d(\C) \to \M_d(\C)$, $\Phi(X) = U X U^*$. It is easy to check that the spectrum of $\Phi$ is the set 
\[\{\lambda_1 \overline \lambda_2\ \, | \, \lambda_1, \lambda_2 \text{ eigenvalues of } U\}.\] 
Since $\Phi$ maps pure states (i.e. rank-one projectors) to pure states, it neither irreducible, nor strictly positive.
\end{eg}

Obviously, a strictly positive map is irreducible. In fact, the following characterization of irreducibility is known \cite{evans}.

\begin{prop}
A positive linear map $\Phi:\M_d(\C) \to \M_d(\C)$ is irreducible if and only if the map $(1+\Phi)^{d-1}$ is strictly positive.
\end{prop}

Irreducible unital maps which satisfy the Schwarz inequality have very nice peripheral spectra. The proof of the following important result can be found in one of \cite{evans, farenick, groh}, in more general settings. 
\begin{thm}

If $\Psi$ is a unital, irreducible map on $\M_d(\C)$ which satisfies the Schwarz inequality, then the set of peripheral (i.e. modulus one) eigenvalues is a (possibly trivial) subgroup of the unit circle $\T$. Moreover, every peripheral eigenvalue is simple and the corresponding eigenspaces are spanned by unitary elements of $\M_d(\C)$.

\end{thm}

Irreducible (and, in particular, strictly positive) quantum channels have desirable spectral properties, hence the interest one has for these classes of maps. As we shall see in Section \ref{sec:fixed}, irreducible maps are in certain sense generic. On the other hand, the strict positivity condition is rather restrictive and not suitable for the considerations on this work. Next, we develop these ideas, giving criteria for irreducibility and for strict positivity.

Let us start by analyzing strict positivity. Subspaces of product spaces $\C^{d} \otimes \C^{d'}$ with high entanglement have received recently great attention. In this direction, applications to the additivity conjecture \cite{hayden-leung, hayden-winter} are the most notable ones. The results in these papers, which rely on probability theory techniques deal with von Neumann entropy. When one looks at the rank, projective algebraic geometry comes into play. Indeed, possible states of the coupled system are modeled by the projective space $\P^{dd'}$. This space contains the \emph{product states}, $\P^{d-1} \otimes \P^{d'-1}$ as a subset called \emph{the Segre variety}. The following lemma, a textbook result in algebraic geometry, is obtained by computing the dimension of the Segre variety (see \cite{cubitt, partha, walgate}).

\begin{lem}\label{lem:entangled_subspace}
The maximum dimension of a subspace $S \subset \C^{d} \otimes \C^{d'}$ which does not contain any non-zero product elements $x \otimes y$ is $(d-1)(d'-1)$. 
\end{lem}

As a rather simple consequence of this lemma, we obtain a necessary condition for strict positivity.

\begin{prop}
Let $\Phi : \M_d(\C) \to \M_d(\C)$ be strictly positive quantum map. Then the Choi rank of $\Phi$ is at least $2d-1$.
\end{prop}
\begin{proof}
Let $\Phi(X) = \sum_{i=1}^k L_i X L_i^*$ be a minimal Kraus decomposition of a strictly positive channel $\Phi$. 
For all $x \neq 0$, $\Phi(\ketbra x x)$ has full rank, and thus, for all non-zero $y \in \C^d$, 
\[\trace\left[\Phi(\ketbra x x) \ketbra y y \right] = \sum_{i=1}^k \abs{\scalar{y}{L_i x}}^2 >0.\]
Hence, for all non-zero $x,y \in \C^d$, there exist an $i$ such that $\scalar{y}{L_i x} = \trace[L_i \ketbra x y] \neq 0$, or, in other words, $L_i^*$ is \emph{not} orthogonal to $\ketbra x y$ with respect to the Hilbert-Schmidt scalar product. Consider now the space $S = \bigcap_{i=1}^k (L_i^*)^\orth \subset \M_d(\C)$. Obviously, $S$ does not contain any rank one matrices $\ketbra{x}{y}$. Under the usual isomorphism $\C^d \otimes (\C^d)^* \simeq \M_d(\C)$, product vectors $x \otimes y$ are identified with rank one matrices $\ketbra{x}{y}$, so, by the Lemma \ref{lem:entangled_subspace}, we get $\dim S \leq (d-1)^2 = d^2-(2d-1)$. Since $S$ is the intersection of $k$ subspaces of dimension $d^2-1$, we get $d^2-k \leq \dim S \leq d^2 - (2d-1)$ which implies $k \geq 2d-1$.
\end{proof}

We now turn to irreducible quantum maps and state some results which will be useful later, when showing that irreducibility is generic for a specific model of random quantum channels. 

The following result of \cite{farenick} gives necessary and sufficient conditions for a map written in the Kraus form to be irreducible. We denote by $\Lat(T)$ the lattice of invariant subspaces of an operator $T \in \M_d(\C)$.

\begin{prop}\label{prop:irred_lat}
Consider the map $\Phi:\M_d(\C) \to \M_d(\C)$ defined by $\Phi(x) = \sum_{i=1}^k L_i X L_i^*$, with $L_i \in \M_d(\C)$, $i=1, \ldots, k$. Then $\Phi$ is irreducible if and only if $\bigcap_{i=1}^k \Lat(L_i)$ is trivial.
\end{prop}

Of course, quantum channels of Choi rank one (i.e. unitary conjugations, see also Example \ref{eg:unitary_conj}), $\Phi(X) = L X L^*$, with $L^*L=\I$ cannot be irreducible, since they leave invariant eigenprojectors of $L$. When looking at channels with Choi rank at least two, an useful criterion for deciding whether $\bigcap_{j=1}^k \Lat(L_j)$ is trivial or not is given by the following two results. The first proposition gives necessary and sufficient conditions for two matrices $A$ and $B$ to share a common eigenvector, and the second one generalizes this idea to arbitrary common subspaces. 

\begin{prop}[The Shemesh criterion, \cite{shemesh}]
Two matrices $A,B \in \M_d(\C)$ have a common eigenvector if and only if
\[ \bigcap_{i,j=1}^{d-1} \ker [A^i,B^j] \neq \{0\},\]
or, equivalently, iff
\[\det \sum_{i,j=1}^{d-1} [A^i,B^j]^* \cdot [A^i,B^j] = 0.\]
\end{prop}

In order to move on from common eigenvectors to common invariant subspaces, we consider antisymmetric tensor powers (or wedge powers) of matrices (see \cite{bhatia}, Ch. I). Given $A \in \M_d(\C)$ and and an integer $1 \leq k \leq n$, the $k$-th wedge power of $A$, denoted by $A^{\wedge k}$, is defined as the restriction of $A^{\otimes k}$ to the antisymmetric tensor product $(\C^d)^{\wedge k}$. More precisely, $A^{\wedge k}$ is a $n \times n$ matrix, where $n=\binom{d}{k}$. Its matrix elements are indexed by couples $(\alpha, \beta)$ of strictly increasing sequences of size $k$ from $\{1, \ldots , d\}$:
\[\left(A^{\wedge k}\right)_{\alpha, \beta} = \det A[\alpha | \beta],\]
where $A[\alpha | \beta]$ is the submatrix of $A$ with rows indexed by $\alpha$ and columns indexed by $\beta$. The next result of \cite{george-ikramov} is an easy consequence of the fact that if $\lambda_1, \ldots, \lambda_k$ are eigenvalues of $A$ with linear independent vectors $v_1, \ldots, v_k$, then $\lambda_1 \lambda_2 \cdots \lambda_k$ is an eigenvalue of $A^{\wedge k}$ with corresponding eigenvector $v_1 \wedge \cdots \wedge v_k$.

\begin{prop}[Generalized Shemesh criterion, \cite{george-ikramov}]\label{prop:shemesh_gen}
Let $A, B \in \M_d(\C)$ be two complex matrices. If $A$ and $B$ have a common invariant subspace of dimension $k$ (for $1 \leq k \leq d-1$), then their $k$-th wedge powers have a common eigenvector, and hence (we put $n=\binom{d}{k}$)
\[ \bigcap_{i,j=1}^{n-1} \ker [(A^{\wedge k})^i,(B^{\wedge k})^j] \neq \{0\},\]
or, equivalently,
\[\det \sum_{i,j=1}^{n-1}  [(A^{\wedge k})^i,(B^{\wedge k})^j]^* \cdot [(A^{\wedge k})^i,(B^{\wedge k})^j] = 0.\]
\end{prop}

%
\begin{rem}
The preceding conditions turn out to be sufficient under more stringent assumptions on the matrices $A$ and $B$ (see \cite{george-ikramov} for further details).
\end{rem}

The main point of the two preceding results is that there exists an universal polynomial $P \in \R[X_1, \ldots, X_{4d^2}]$ with the property that whenever two matrices $A=(a_{ij})$ and $B=(b_{kl})$ have a non-trivial common invariant subspace, $P(\Re a_{ij},\Im a_{ij},\Re b_{kl}, \Im b_{kl}) = 0$. This fact (together with Proposition \ref{prop:irred_lat}) will be useful later in this work, when we shall show that a generic class of quantum maps are irreducible.

\section{Non-random repeated interactions and a new model of random density matrices}\label{sec:fixed}

In this section we consider repeated interactions with a fixed unitary matrix $U$ ($\forall n, U_n = U$) and fixed state of the environment $\beta$ ($\forall n, \beta_n = \beta$). By the results of the previous section, the recurrence relation which governs the discrete, deterministic dynamics is
\[\rho_{n+1}=\Phi(\rho_n) = \trace_{\K}\left[U(\rho_n\otimes\beta)U^*\right].\]
Iterating this formula, one obtains the state of the system after $n$ interactions:
\[\rho_n = \Phi^n(\rho_0),\]
where $\rho_0$ was the initial state of the system. There is one obvious situation in which the asymptotic properties of the sequence $(\rho_n)_n$ can be established. Indeed, from Lemma \ref{lem:channel}, one knows that all quantum channels have eigenvalue 1 and that all other eigenvalues have module less than 1. Let $\mathcal C$ be the set of all quantum channels that have 1 as a simple eigenvalue and all other eigenvalues are contained in the \emph{open} unit disc. Since 1 is a simple eigenvalue, $\Phi$ has an unique fixed point which is (by  Lemma \ref{lem:channel}) a density matrix $\rho_\iy \in \MD_d(\C)$. Using the Jordan form of $\Phi$, one can show the following result (\cite{terhal}). 
\begin{prop}\label{prop:conv_1_unique}
Let $\Phi \in \mathcal C$ be a fixed quantum channel. Then, for all $\rho_0 \in \MD_d(\C)$,
\[\lim_{n \to \iy} \Phi^n(\rho_0) = \rho_\iy,\]
where $\rho_\iy$ is the unique invariant state of $\Phi$.
\end{prop}

The importance of the peripheral spectrum of a quantum channel is illustrated in the following example. 
\begin{eg}\label{eg:periph}
Consider the following channel $\Phi : \M_2(\C) \to \M_2(\C)$
\[\Phi(X) = \frac 1 2 \sigma_1 X \sigma_1 + \frac 1 2 \sigma_3 X \sigma_3,\]
where the Pauli matrices are given by
\[\sigma_0 = \I =\begin{bmatrix}
1&0\\
0&1
\end{bmatrix}, \quad \sigma_1=\begin{bmatrix}
0&1\\
1&0
\end{bmatrix}, \quad \sigma_2=\begin{bmatrix}
0&-i\\
i&0
\end{bmatrix}, \quad \sigma_3=\begin{bmatrix}
1&0\\
0&-1
\end{bmatrix}.\]
Direct computation shows that $\Phi(\I) = \I$, $\Phi(\sigma_2) = -\sigma_2$ and $\Phi(\sigma_1) = \Phi(\sigma_3) = 0$. Hence, the peripheral spectrum of $\Phi$ has 2 simple eigenvalues, $1$ and $-1$. However, for $\rho_0 = 1/2 (\I + \sigma_2) \in \MD_2(\C)$, one has
\[ \Phi^n(\rho_0) = \frac 1 2 \left(\I + (-1)^n \sigma_2 \right),\]
which does not converge in the limit $n \to \iy$. Hence, the simplicity of the eigenvalue $1$ does not suffice to have convergence to the invariant state. Note also that the channel $\Phi$ is irreducible, since $\sigma_1$ and $\sigma_3$ do not have any common non-trivial invariant subspaces.
\end{eg}
%
%
%

Let us now show that the class $\mathcal{C}$ of quantum channels which have 1 as an unique peripheral eigenvalue is generic in a certain sense. To this end, we shall introduce a model of \emph{random} quantum channel, based on the Stinespring decomposition. To start, fix the dimension of the environment $d'$ and a state $\beta \in \MD_{d'}(\C)$. Next, consider an unitary random matrix $U$ distributed along the (uniform) Haar measure $\Haar_{dd'}$ on $\U(dd')$. To the state $\beta$ and the evolution operator $U$, we associate the quantum channel $\Phi_{U,\beta}$. In this way, we define a model of random quantum channels by considering the image measure of the Haar distribution $\Haar_{dd'}$ on the set of quantum channels. In the recent preprint \cite{zyc_rqo}, the authors study a similar model of random quantum channels, focusing on the spectral properties of the random matrix defining the channel.

More precisely, we claim that if the state of the environment $\beta$ is fixed and the interaction unitary $U \in \U(dd')$ is chosen randomly with the uniform Haar distribution $\Haar_{dd'}$, then, with probability one, the channel $\Phi^{U,\beta}$ admits 1 as the unique eigenvalue on the unit circle. Here we need another fact from algebraic geometry, summarized in the following lemma (for a similar result, one should have a look at Proposition 2.6 of \cite{arveson}).

\begin{lem}\label{lem:unitary_poly}
Given a polynomial $P \in \R[X_1, \ldots, X_{2d^2}]$, the set 
\[Z = \{U=(u_{ij}) \in \U(d) \, | \, P(\Re u_{ij}, \Im u_{ij})=0\}\]
is either equal to the whole set $\U(d)$ or it has Haar measure 0.
\end{lem}
\begin{proof}
We start by noticing that the real algebraic set $\U(d)$ is irreducible. This follows from the connectedness of $\U(d)$ (in the usual topology) and from the fact that irreducible components of a linear algebraic group are disjoint (\cite{humphreys}, 7.3). The set $Z$ is the intersection of the irreducible variety $\U(d)$ with the variety $V$ of zeros of the polynomial $P$. If $\U(d) \subset V$, then $Z = \U(d)$; otherwise, the dimension of $Z$ is strictly smaller than $d^2$, the real dimension of $\U(d)$. Since the Haar measure is just the integration of an invariant differential form, it has a density in local coordinates (\cite{faraut}, Ch. 5) and hence $\Haar_d(Z) = 0$ in this case.
\end{proof}

\begin{thm}\label{thm:haar_1_unique}
Let $\beta$ be a fixed density matrix of size $d'$. If $U$ is a random unitary matrix distributed along the Haar invariant probability $\Haar_{dd'}$ on $\U(dd')$, then $\Phi^{U, \beta} \in \mathcal C$ almost surely.
\end{thm}
\begin{proof}
The proof goes in two steps. First, we show that $\Phi^{U, \beta}$ is almost surely irreducible and then we conclude by a simple probabilistic argument. 

Let us start by applying Lemma \ref{lem:unitary_poly} to show that a random quantum channel is almost surely irreducible. To this end, using Eq.~(\ref{eq:Kraus_from_U}), we obtain a set of Kraus operators for $\Phi^{U, \beta}$ which are sub-matrices of $U \in \U(dd')$. Consider two such Kraus operators $A, B \in \M_d(\C)$ (choose $j$ such that $b_j \neq 0$ and take $A = U_{1j}$, $B=U_{2j}$). Using Proposition \ref{prop:irred_lat}, to show irreducibility it suffices to see that $A$ and $B$ do not have a non-trivial common invariant subspace. Let $1 \leq k \leq d-1$ be the dimension of a potentially invariant common subspace of $A$ and $B$. By the criterion in Proposition \ref{prop:shemesh_gen}, there exists a polynomial $P_k$ in the entries of $A$ and $B$ (and thus in the entries of $U$) such that if $P_k(U)$ is non-zero, then $A$ and $B$ do not share a $k$-dimensional invariant space. Note that $P_k$ can not be identically zero: for two small enough matrices $\tilde A, \tilde B$ without common invariant subspaces, one can build a unitary matrix $\tilde U$ such that $\tilde A = \tilde U_{1j}$, $\tilde B=\tilde U_{2j}$. By the Lemma \ref{lem:unitary_poly}, $\Haar_{dd'}$-almost all unitary matrices $U$ give Kraus operators $A$ and $B$ that do not have any $k$-dimensional invariant subspaces in common. Since the intersection of finitely many full measure sets has still measure one, almost all quantum channels are irreducible.

Consider now a random channel $\Phi^{U, \beta}$ which we can assume irreducible. Since the peripheral spectrum of an irreducible channel is a multiplicative subgroup of the unit circle $\T$, it suffices to show that for all element $\lambda$ of the finite set $\{\xi \in \T | \exists 1 \leq n \leq d^2 \text{ s.t. } \xi^n = 1\} \setminus \{1\}$, with Haar probability one, $\lambda$ is not an eigenvalue of $\Phi^{U, \beta}$. We use the same trick as earlier. Consider such a complex number $\lambda$ and introduce the polynomial $Q_\lambda(U) = \det [\Phi^{U, \beta} - \lambda \I_{(dd')^2}]$, where $\Phi^{U, \beta}$ is seen as a matrix $\Phi^{U, \beta} \in \M_{(dd')^2}(\C)$. Since $\lambda \neq 1$ and the identity channel $\Phi^{U=\I, \beta}$ has only unit eigenvalues, $Q_\lambda(U)$ cannot be identically zero, and the conclusion follows.
\end{proof}

\begin{rem}
The main difficulty in the proof of the preceding result comes from the fact that the matrices $A$ and $B$ are ``correlated'': two blocks of an unitary matrix must satisfy norm and (maybe) orthogonality relations. Hence the need to use sophisticated geometric algebra techniques. Proving that two independent random Gaussian (or unitary) matrices do not share non-trivial invariant subspaces is much simpler and does not require the use of such techniques.
\end{rem}

We now move on and apply the previous results to constructing a new family of probability distributions on the set of density matrices. The main idea is to assign, whenever possible, to a random unitary $U \in \U(dd')$ its unique invariant density matrix $\rho_\iy$. In this way, the Haar measure $\Haar_{dd'}$ on the unitary group $\U(dd')$ is transported to the set of density matrices $\MD_d(\C)$.

Let us now make this construction more precise. The new family of probability measures shall be indexed by an integer $d' \geq 1$ (the dimension of the auxiliary system) and by a non-increasing probability vector $b=(b_1, \ldots, b_{d'}) \in \C^{d'}$: $b_1 \geq b_2 \geq \ldots \geq b_{d'} \geq 0$ and $\sum_i b_i = 1$ (these are the eigenvalues of the state of the auxiliary system). For such a couple $(d', b)$ consider a density matrix $\beta \in \MD_{d'}(\C)$ with eigenvalue vector $b$ (the eigenvectors of $\beta$ do not matter, see Lemma \ref{prop:asympt_measures}). As it follows from Proposition \ref{thm:haar_1_unique}, for almost all unitaries $U \in \U(dd')$, the channel $\Phi^{U, \beta}$ satisfies the hypotheses of Proposition \ref{prop:conv_1_unique}. Hence, for almost all $U$ and for all density matrices $\rho_0 \in \MD_d(\C)$, $\lim_{n \to \iy}(\Phi^{U, \beta})^n\rho_0 = \rho_\iy$, where $\rho_\iy$ is the unique invariant state of $\Phi^{U, \beta}$. We have defined almost everywhere an application 
\begin{align*}
\U(dd') &\to \MD_d(\C)\\
U &\mapsto \rho_\iy.
\end{align*}
We denote by $\nu_{b}$ the image measure of the Haar probability $\Haar_{dd'}$ on $\U(dd')$ by the previous application (notice that we dropped the integer parameter $d'$, since this is the dimension of the vector $b$). We call $\nu_b$ the \emph{asymptotic induced measure} on the set of density matrices.

We now motivate the term ``asymptotic induced'' in the previous definition by showing how the measures $\nu_b$ relate to the induced random density matrices considered in \cite{zyc_sommers, nechita}. Let us recall here how these measures are constructed and how one can sample from this distribution. The physical motivation behind the induced measures comes from the following setup. Assume that a system $\S$ is coupled to an environment $\mathcal E$ and that the whole is in a pure state $\psi \in \H \otimes \K$. If one has no \emph{a priori} knowledge about the state $\psi$, then it is natural to assume that $\psi$ is a random uniform element on the unit sphere of the product space $\H \otimes \K$. The distribution of the partial trace over the environment 
\[ \rho_1 = \trace_\K[\ketbra{\psi}{\psi}]\]
is called the \emph{induced measure} and it is denoted by $\mu_{d'}$ (the parameter $d' = \dim \K$ is the dimension of the environment). We refer the interested reader to \cite{nechita} for more information on these measures. Since the distribution of an uniform norm-one vector $\psi$ is the equal to the distribution of $U \psi_0$, where $\psi_0$ is any fixed norm-one vector and $U$ is a Haar unitary, $\mu_{d'}$ is also the distribution of the matrix
\[ \rho_1 = \trace_\K[U\ketbra{\psi_0}{\psi_0}U^*].\]
If one chooses $\psi_0 = e_1 \otimes f_1$, where $e_1$ and $f_1$ are the first vectors of the canonical basis of $\C^d$ and respectively $\C^{d'}$, then
\[ \rho_1 = \trace_\K[U(\rho_0 \otimes \beta_0)U^*] = \Phi^{U, \beta_0}(\rho_0),\]
with $\rho_0 = \ketbra{e_1}{e_1}$ and $\beta_0 = \ketbra{f_1}{f_1}$. Hence, the induced measure $\mu_{d'}$ is the distribution of the result of \emph{one} application of a random channel $\Phi^{U, \beta_0}$ on the constant matrix $\rho_0$: $\rho_1 \sim \mu_{d'}$. On the other hand, after a large number of identical interactions, one gets
\[ \rho_\iy = \lim_{n \to \iy} \left[\Phi^{U, \beta_0}\right]^n(\rho_0).\]
In this work we have shown that with $\Haar_{dd'}$-probability one, $\rho_\iy$ is a well defined density matrix-valued random variable which does not depend on the value of $\rho_0$. Since the eigenvalue vector of $\beta_0$ is $b_0 = (1, 0, \ldots, 0) \in \C^{d'}$ we have that $\rho_\iy \sim \nu_{b_0}$. Now, the relation between the two families of measures is clear: the induced measure $\mu_{d'}$ is the distribution of the density matrix after one interaction, whereas $\nu_{b_0}$ is the distribution at the limit, after a large number of interactions. The reader may notice that this analogy is valid only in the case where $b=(1, 0, \ldots, 0)$ (pure state on the environment). Generalizations of the (usual) induced measures to other environment states are possible, but out of the scope of the present work. To further compare the asymptotic and the one interaction induced measures, we plotted the spectra of samples of density matrices from both families in Figures \ref{fig:asympt} and \ref{fig:induced}. In particular, one should compare Figure \ref{fig:asympt:22} with Figure \ref{fig:induced:22} ($d=d'=2$), Figure \ref{fig:asympt:33} with Figure \ref{fig:induced:33} ($d=d'=3$) and Figure \ref{fig:asympt:35} with Figure \ref{fig:induced:35} ($d=3, d'=5$).

\begin{figure}[htb]
\centering
\subfigure[]{\label{fig:asympt:22}\includegraphics[width=4.5cm]{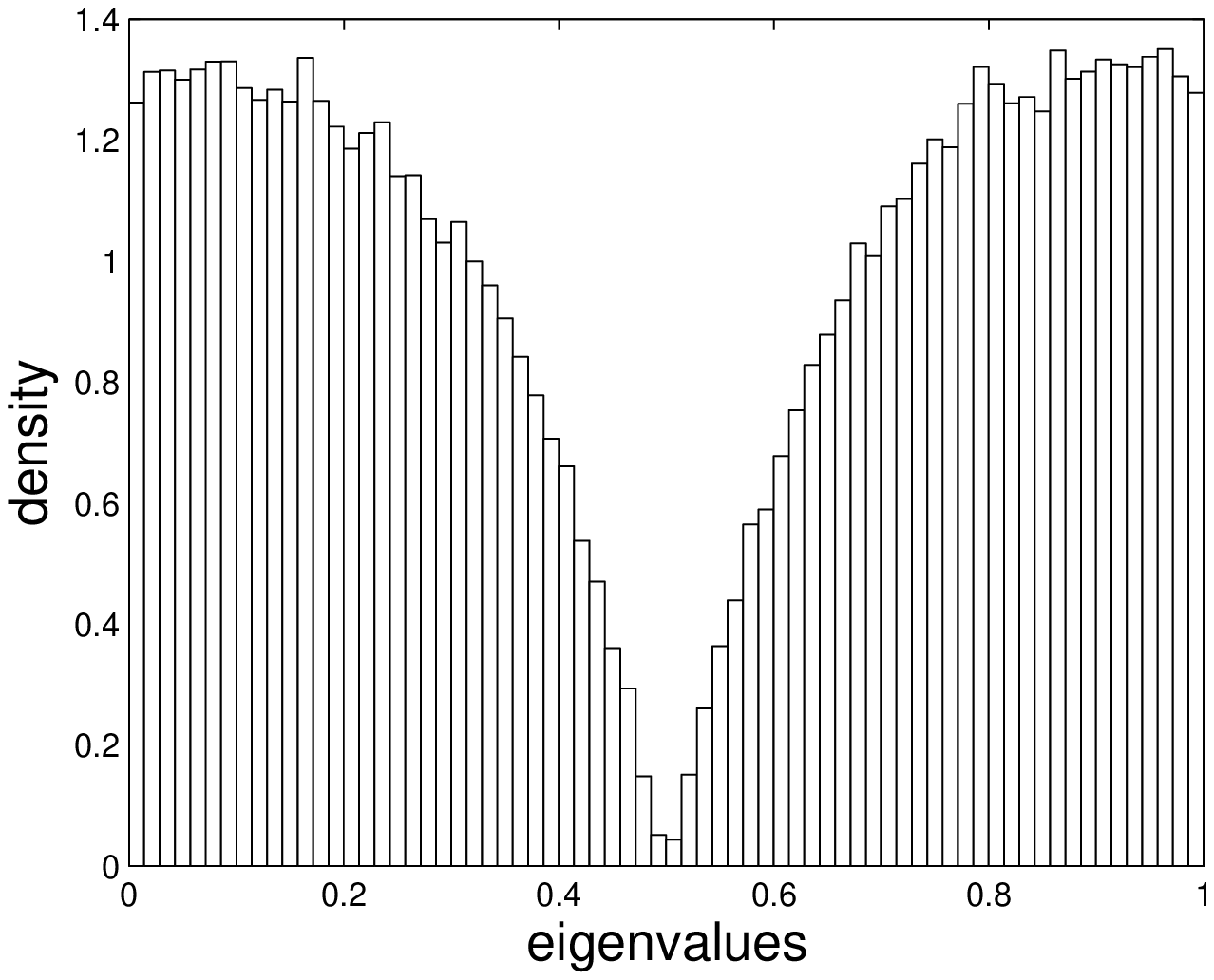}}
\subfigure[]{\includegraphics[width=4.5cm]{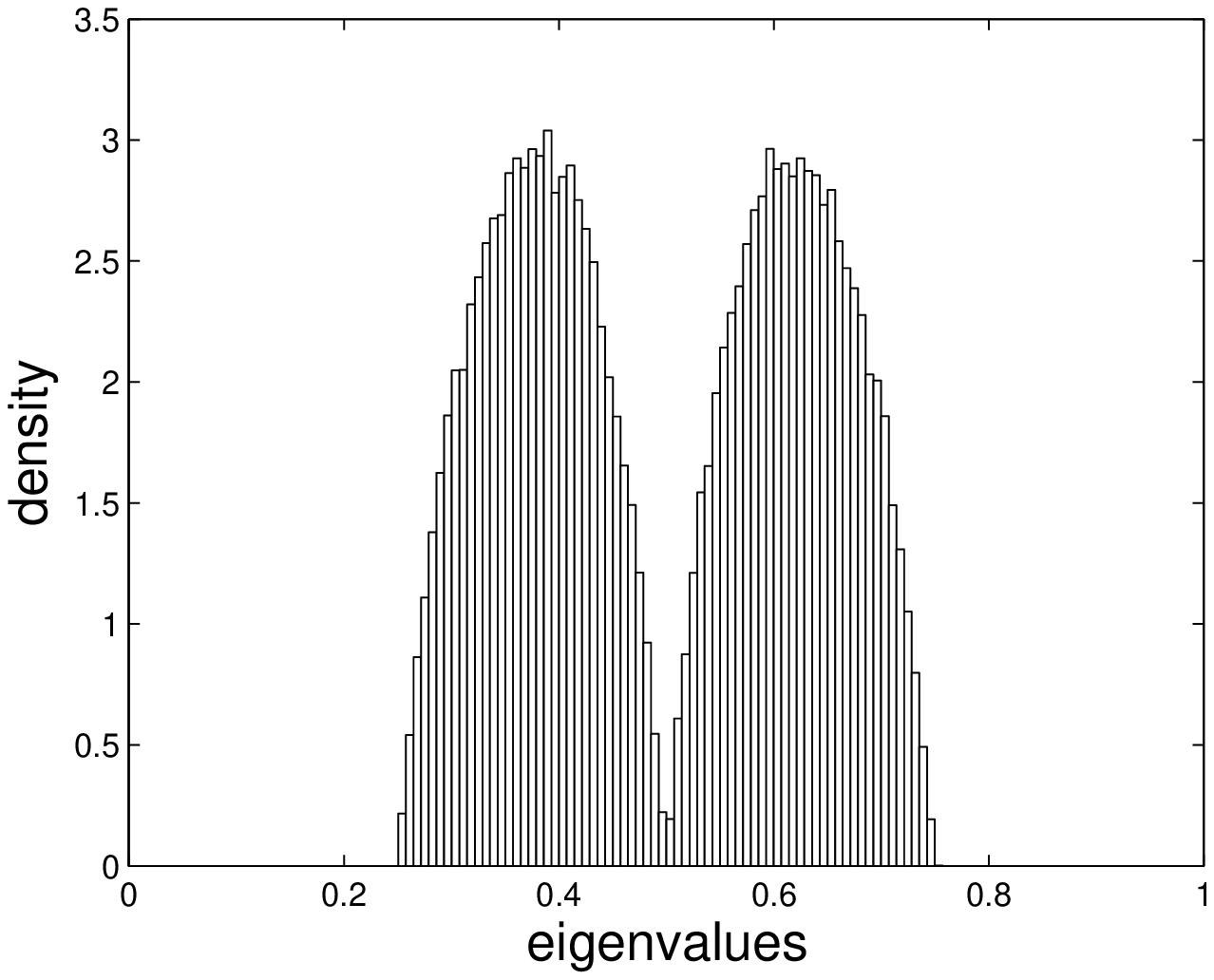}}
\subfigure[]{\includegraphics[width=4.5cm]{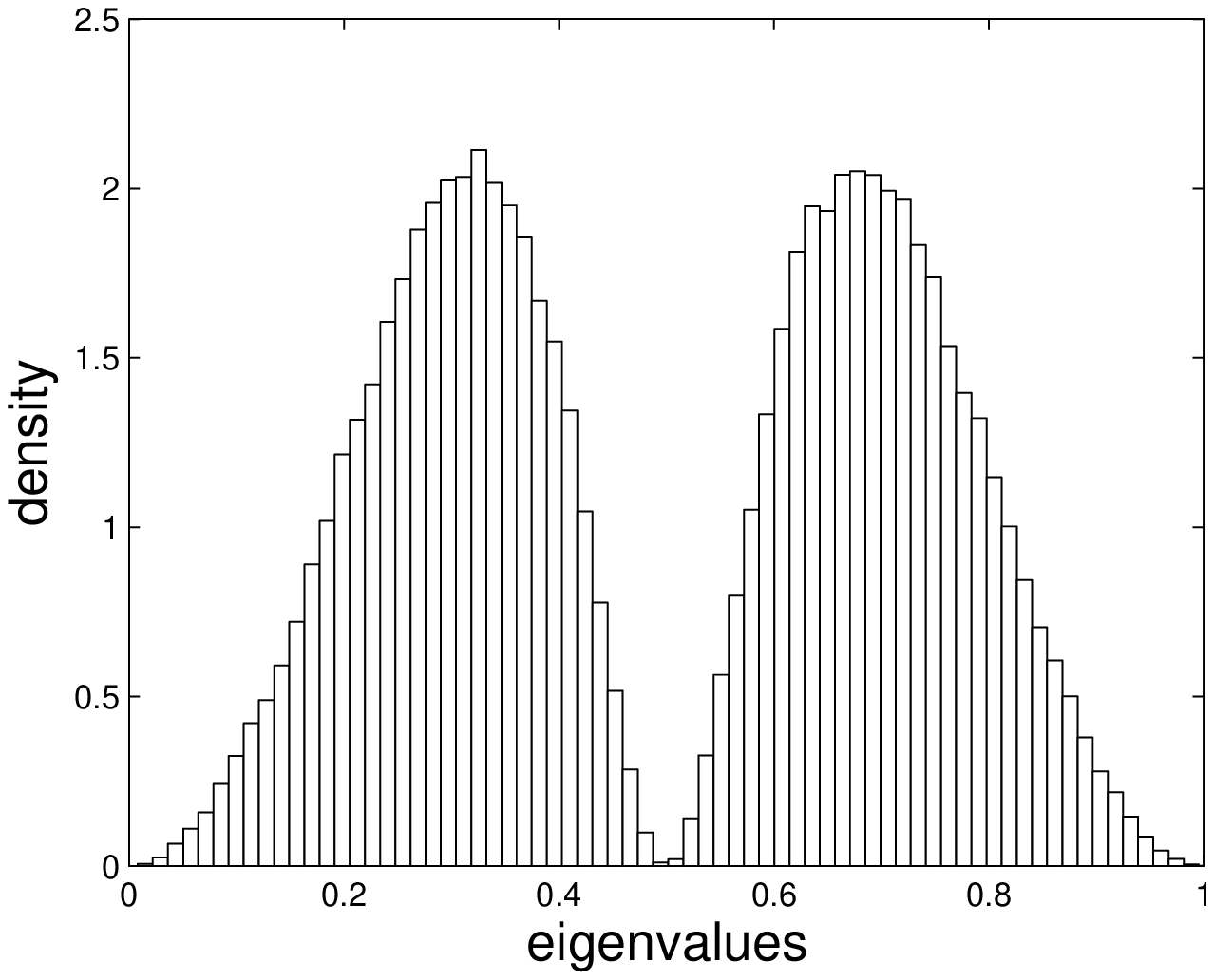}}\\
\subfigure[]{\label{fig:asympt:33}\includegraphics[width=4.5cm]{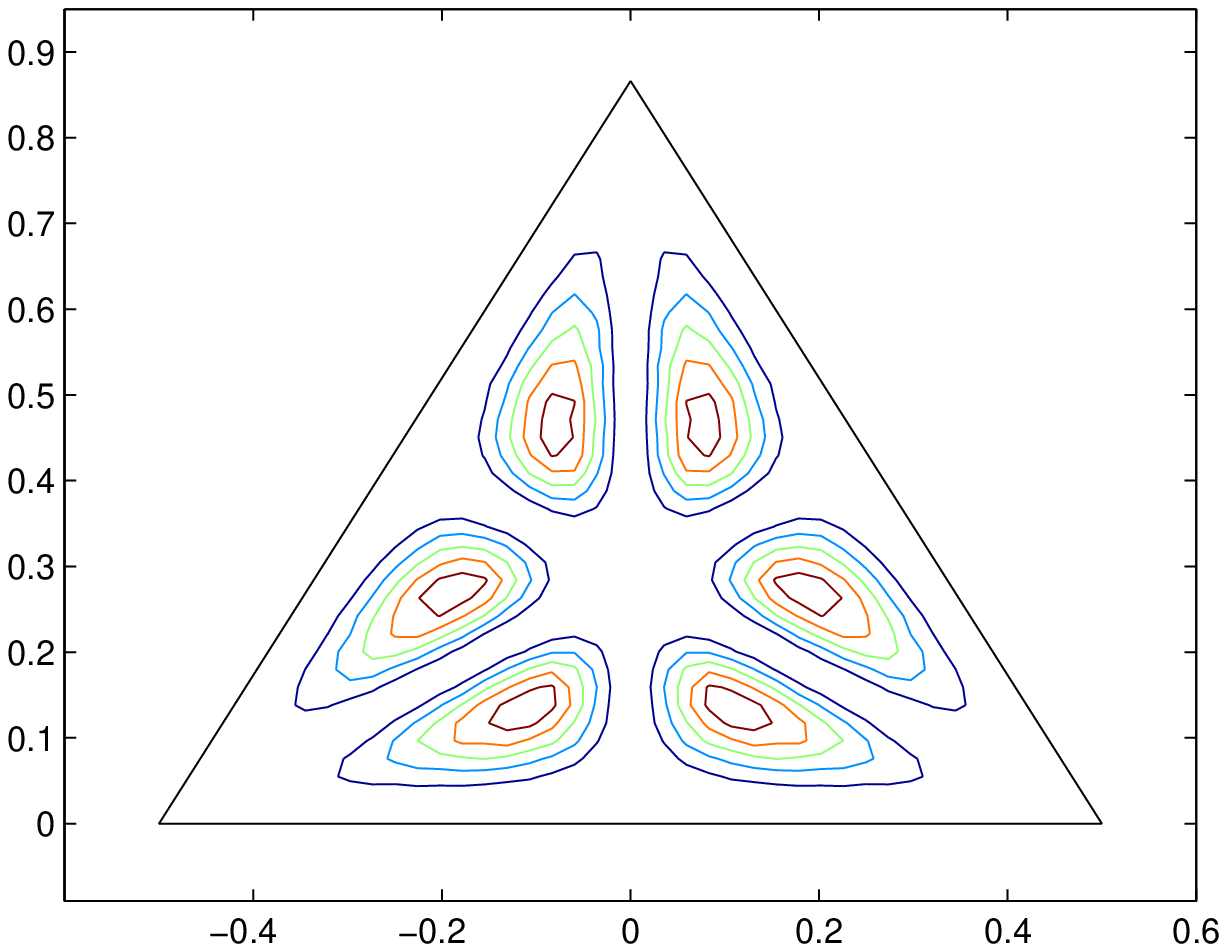}}
\subfigure[]{\includegraphics[width=4.5cm]{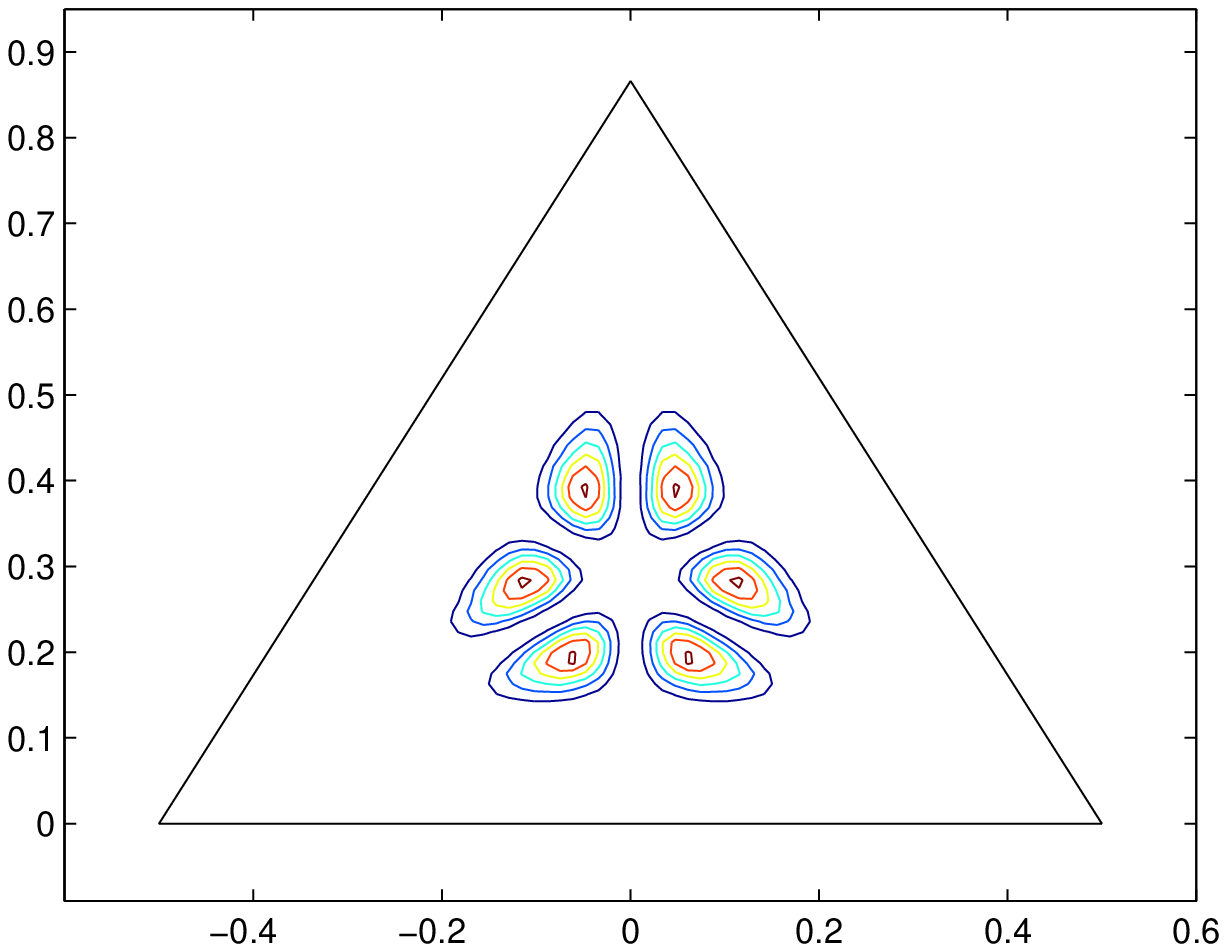}}
\subfigure[]{\label{fig:asympt:35}\includegraphics[width=4.5cm]{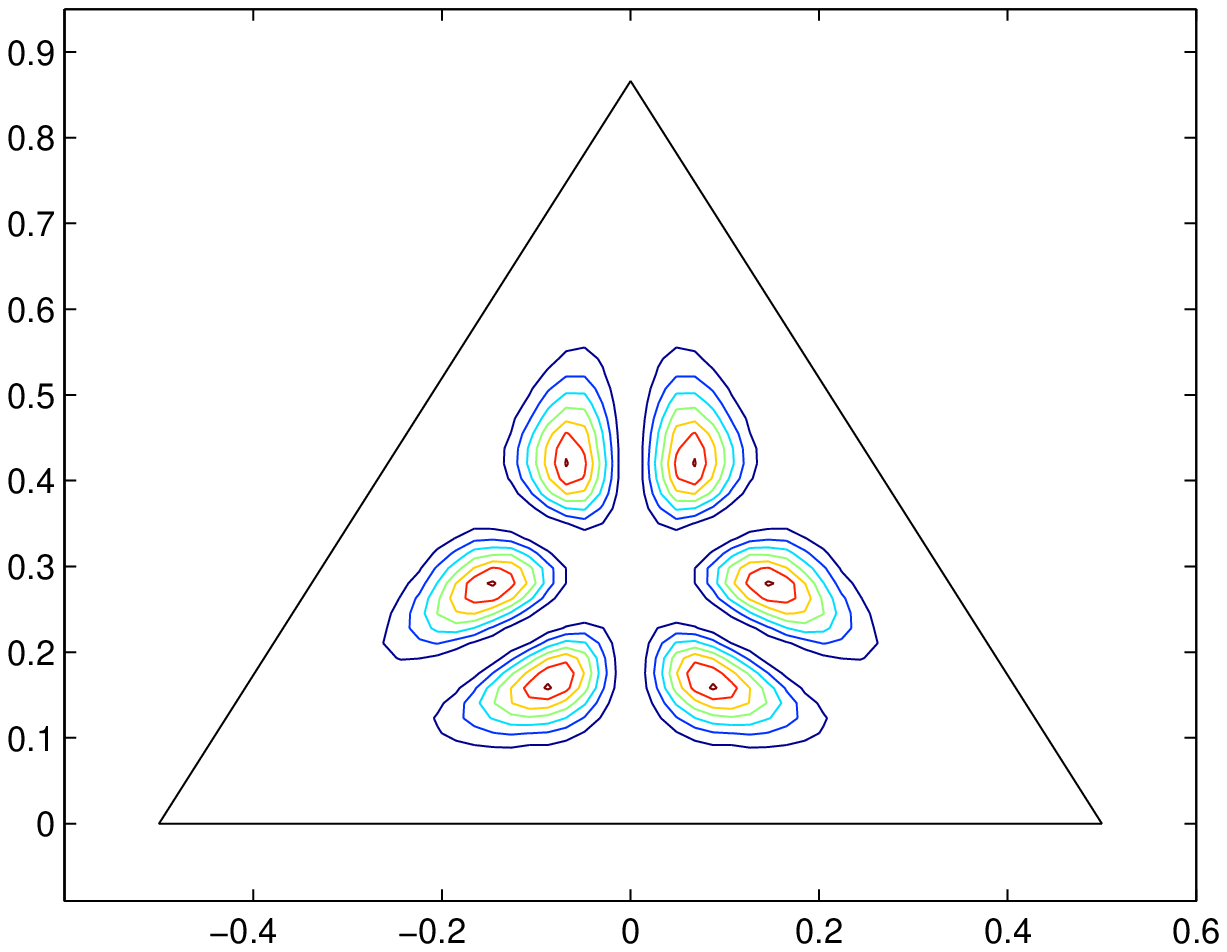}}
\caption{Asymptotic measure eigenvalue distribution. First row, from left to right: $(d=2,b=[1, 0])$, $(d=2,b=[3/4, 1/4])$, $(d=2,b=[1, 0, 0, 0])$. Second row: $(d=3,b=[1, 0, 0])$, $(d=3,b=[3/4, 1/8, 1/8])$ and $(d=3,b=[1, 0, 0, 0, 0])$.}
\label{fig:asympt}
\end{figure}

\begin{figure}[htb]
\centering
\subfigure[]{\label{fig:induced:22}\includegraphics[width=4.5cm]{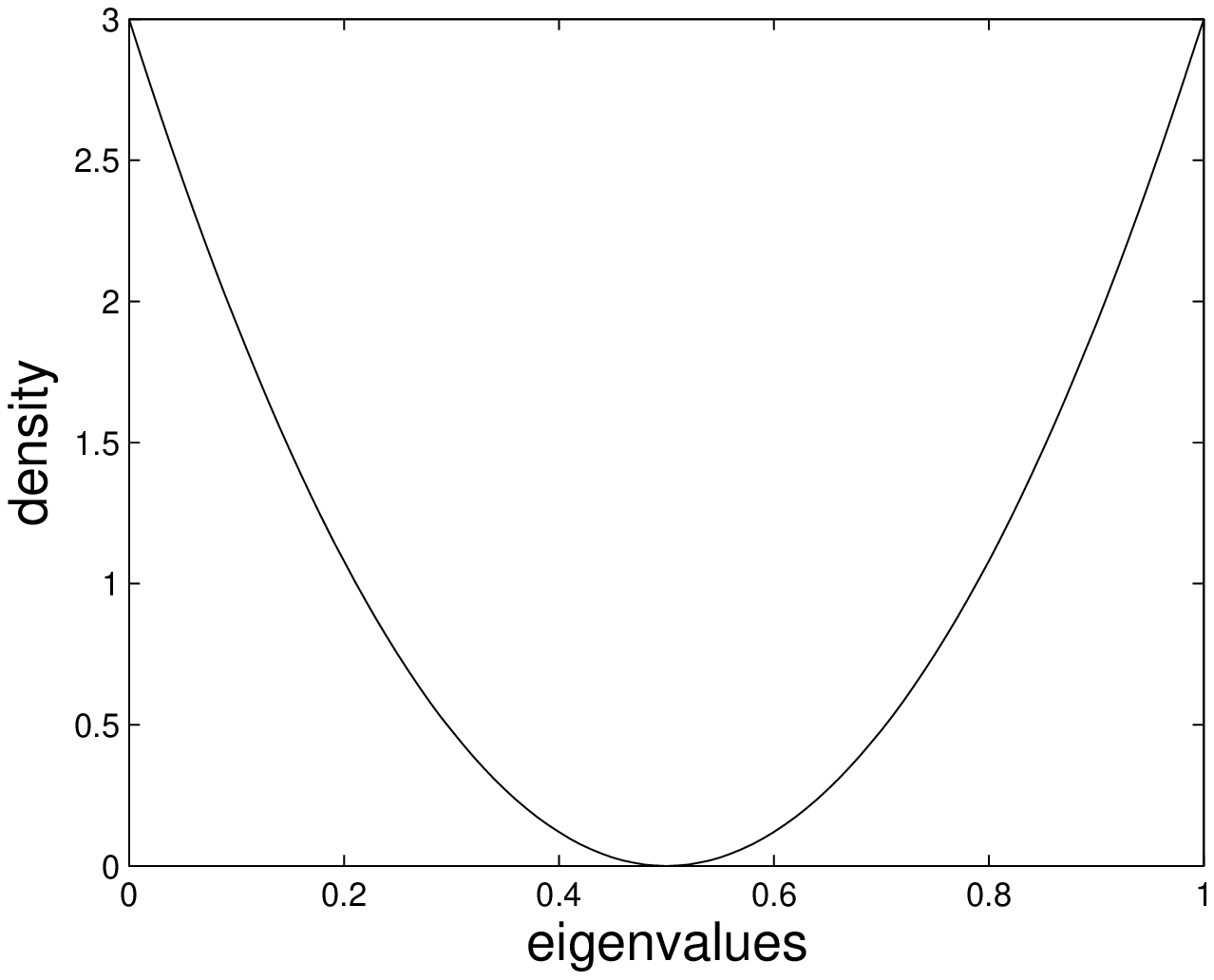}}
\subfigure[]{\label{fig:induced:33}\includegraphics[width=4.5cm]{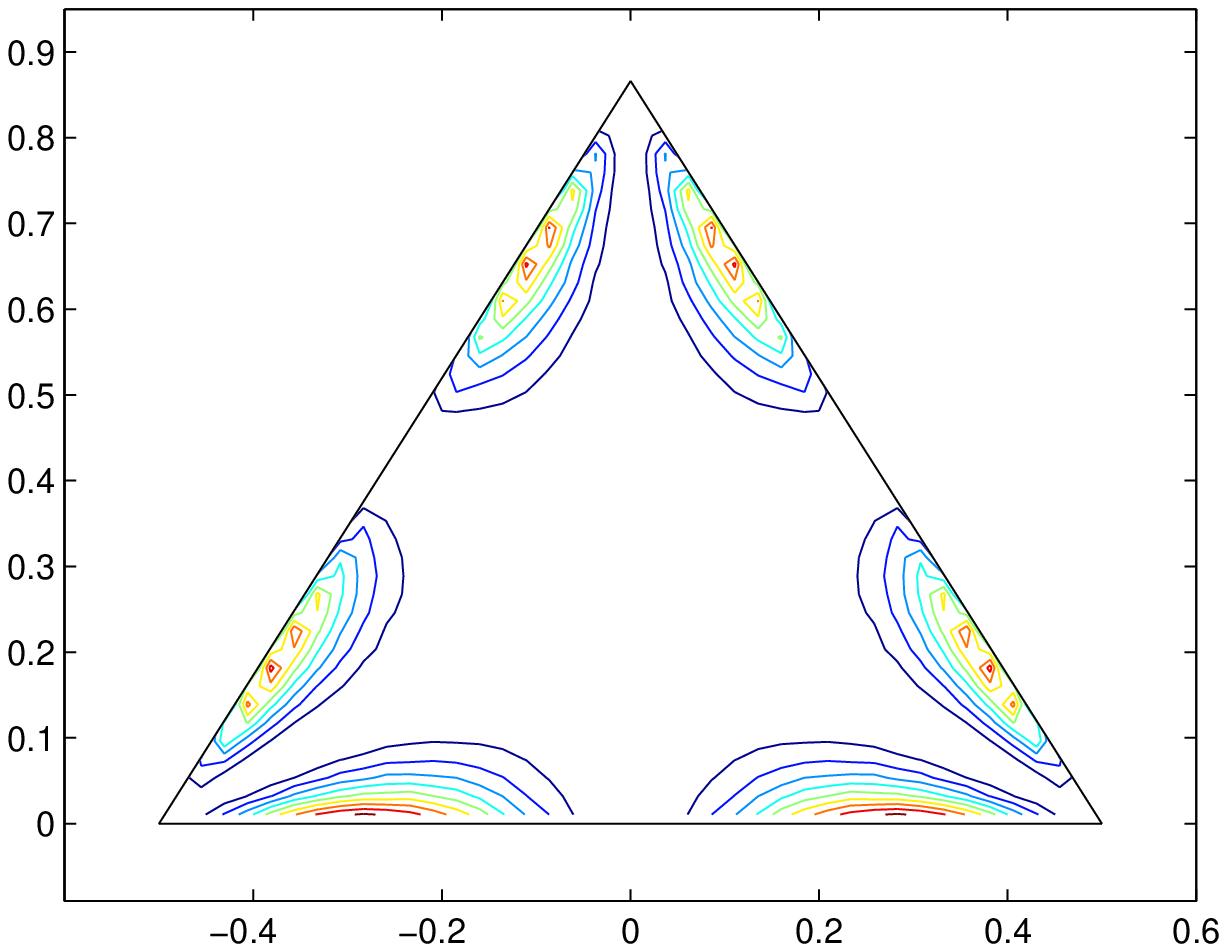}}
\subfigure[]{\label{fig:induced:35}\includegraphics[width=4.5cm]{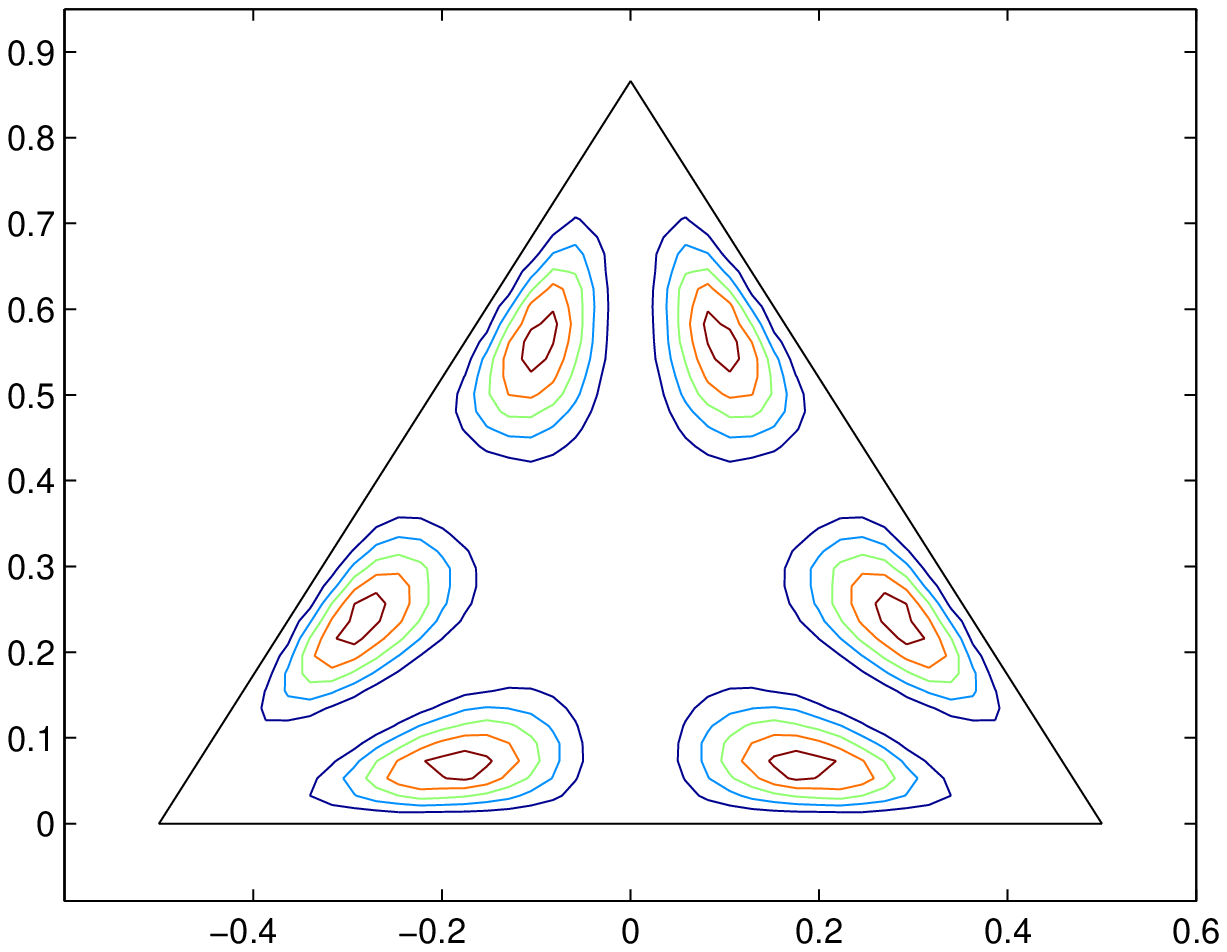}}
\caption{Induced measure eigenvalue distribution for $(d=2,d'=2)$, $(d=3,d'=3)$ and $(d=3,d'=5)$.}
\label{fig:induced}
\end{figure}

One particularly simple case is obtained by taking $b = (1/d', \ldots, 1/d')$. The measure $\nu_b$ is then trivial, being equal to the Dirac mass supported on the ``chaotic state'' $\I/d$. In the next lemma we prove some basic properties of the newly introduced measures $\nu_b$. A more thorough investigation of these measures is postponed to a later work.
\begin{prop}\label{prop:asympt_measures}
The probability measures $\nu_b$ have the following properties:
\begin{enumerate}
\item For every probability vector $b$, the measure $\nu_b$ is well defined, in the sense that the distribution of $\rho_\iy = \lim_{n \to \iy} [\Phi^{U, \beta}]^n(\rho_0)$ does not depend on the eigenvectors of $\beta$, but only on the eigenvalue vector $b$.
\item For all unitary matrix $V \in \U(d)$, $\rho$ and $V \rho V^*$ have the same distribution (we say that the measure $\nu_b$ is \emph{unitarily invariant}). 
\item There exists a probability measure $n_{b}$ on the probability simplex $\Delta_{d-1}$ such that if $D$ is a diagonal matrix sampled from $n_{b}$ and $V$ is an independent Haar unitary on $\U(d)$, then $VDV^*$ has distribution $\nu_{b}$. In other words, the distribution of a random density matrix $\rho \sim \nu_{b}$ is determined by the distribution of its eigenvalue vector $\Delta_{d-1} \ni \lambda \sim n_{b}$.
\end{enumerate}
\end{prop}
\begin{proof}
To prove the first assertion, we show that for all $W \in \U(d')$, replacing $\beta$ with $W \beta W^*$ does not change the distribution of $\rho_\iy$. To see this, note that by the invariance of the Haar probability measure $\Haar_{dd'}$, the random matrices $U$ and $\tilde U = U (\I_d \otimes W)$ have the same distribution. It follows that the same holds for the random channels $\Phi^{U, \beta}$ and $\Phi^{\tilde U, \beta} = \Phi^{U, W \beta W^*}$ and thus for their invariant states. The second affirmation is proved in the same manner (this time using a fixed unitary $V$ acting on $\H$) and the third one is a trivial consequence of the second.
\end{proof}

\section{Repeated interactions with random auxiliary states}\label{sec:random_env}

In the previous section we considered repeated \emph{identical} quantum interactions of a system $\S$ with a chain of identical environment systems $\mathcal E$. We now introduce classical randomness in our model by considering random states on the environment $\mathcal E$. In this model, the unitary describing the interaction is a fixed deterministic matrix $U \in \U(dd')$.

The $n$-th interaction between the small system $\S$ and the environment $\E$ is given by the following relation:
\[\rho_n = \Phi^{\beta_n}(\rho_{n-1}) = \trace_\K [U(\rho_{n-1} \otimes \beta_n)U^*],\]
where $(\beta_n)_n$ is a sequence of independent identically distributed random density matrices. Notice that, since $U$ is constant, we use the shorthand notation $\Phi^\beta = \Phi^{U, \beta}$.

We are interested, as usual, in the limit $n \to \iy$. In this case however, the (random) channels $\Phi^{\beta_n}$ do not have in general a common invariant state, so one has to look at ergodic limits. We use here the machinery developed by L. Bruneau, A. Joye and M. Merkli in \cite{bjm_infprod} (see \cite{bjm_randrep, bjm_asympt} for additional results in this direction). For the sake of completeness, let us state their main result. 

\begin{thm}[\cite{bjm_infprod}, Theorem 1.3.]\label{thm:bjm}
Let $(M_n)_n$ be a sequence of i.i.d. random contractions of $\M_d(\C)$ with the following properties:
\begin{enumerate}
\item There exists a constant vector $\psi \in \C^d$ such that $M(\omega)\psi = \psi$ for (almost all) $\omega$;
\item $\P( \text{the multiplicity of the eigenvalue 1 of } M(\omega) \text{ is exactly one}) >0$.
\end{enumerate}
Then the (deterministic) matrix $\E[M]$ has eigenvalue 1 with multiplicity one and there exists a constant vector $\theta \in \C^d$ such that
\[\lim_{N \to \iy} \frac 1 N \sum_{n=1}^N M_1(\omega)M_2(\omega)\cdots M_n(\omega) = \ketbra{\psi}{\theta} = P_{1, \E[M]},\]
where $P_{1, \E[M]}$ is the rank-one spectral projector of $\E[M]$ corresponding to the eigenvalue 1. 
\end{thm}

Note that this result does't apply to our situation, mainly for two reasons: the order of the composition of the channels $\Phi$ is reversed and the linear applications $\Phi^{\beta_n}$ do not necessarily share a constant invariant state $\psi$. This inconvenient can be overcome by considering dual channels (see Section \ref{sec:spectral}), or, in physicists' language, by switching from the Schr\"odinger to the Heisenberg picture of Quantum Mechanics. Duals of quantum channels are unital, hence they have in common the invariant element $\I$. Another important benefit of considering duals is that the order of composition of maps is reversed. Indeed, if one starts from a state $\rho_0$, applies successively $n$ channels $\Phi_1, \ldots, \Phi_n$ and finally measures an observable $A \in \Msa_d(\C)$, it is easy to see that the expected outcome is 
\[\trace[(\Phi_n \circ \cdots \circ \Phi_1)(\rho) \cdot A] = \trace[(\Phi_{n-1}\circ \cdots \circ \Phi_1)(\rho) \cdot \Psi_n(A)] = \cdots = \trace[\rho \cdot (\Psi_1 \circ \cdots \circ \Psi_n)(A)].\]
We are now in position to state and prove the analogue of Theorem \ref{thm:bjm} for infinite products of quantum channels, simply by replacing quantum channels with their duals.

\begin{thm}\label{thm:prod_q_channels}
Let $(\Phi_n)_n$ be a sequence of i.i.d. random quantum channels acting on $\M_d(\C)$ such that $\P( \Phi \text{ has an unique invariant state}) >0$. Then $\E[\Phi]$ is a quantum channel with an unique invariant state $\theta \in \MD_d(\C)$ and, $\P$-almost surely, 
\[\lim_{N \to \iy} \frac 1 N \sum_{n=1}^N [\Phi_n \circ \cdots \circ \Phi_1](\rho_0) = \theta, \quad \forall \rho_0 \in \MD_d(\C).\]
\end{thm}
\begin{proof}
Let us start by introducing some notation. Let, for some initial state $\rho_0 \in \MD_d(\C)$,
\[ \mu_N = \frac 1 N \sum_{n=1}^N [\Phi_n \circ \cdots \circ \Phi_1](\rho_0),\]
and consider the dual operators $\Psi_n$ which are, as described earlier, the adjoints of $\Phi_n$ with respect to the Hilbert-Schmidt scalar product on $\M_d(\C)$. Then, for a self-adjoint observable $A \in \Msa_d(\C)$, one has
\begin{equation}\label{eq:mu_N_A}
\trace[\mu_N A] = \trace \left[ \rho_0 \frac 1 N \sum_{n=1}^N (\Psi_1 \circ \cdots \circ \Psi_n)(A) \right].
\end{equation} 
It is easy to see that the random operators $\Psi_n$ satisfy the hypotheses of Theorem \ref{thm:bjm} on the Hilbert space $\M_d(\C)$ endowed with the  Hilbert-Schmidt scalar product. Indeed, the spectrum of $\Psi$ is the complex conjugate of the spectrum of $\Phi$, hence $\Psi$ is a contraction (with respect to the Hilbert-Schmidt norm). Moreover, with non-zero probability, $\I_d$ is the unique invariant state of $\Psi$. From the Theorem \ref{thm:bjm}, one obtains the existence of a non-random element $\theta \in \M_d(\C)$ such that, $\P$-almost surely, 
\[\lim_{N \to \iy}\frac 1 N \sum_{n=1}^N \Psi_1 \circ \cdots \circ \Psi_n = \ketbra{\I_d}{\theta}.\]
Plugging this into Eq.~(\ref{eq:mu_N_A}), one gets
\[\lim_{N \to \iy} \trace[\mu_N A] = \trace[\rho_0 \ketbra{\I_d}{\theta}A] = \scalar{\theta}{A}_{\text{HS}} \trace[\rho_0 \I_d] = \trace[\theta^* A].\]
Since the set of density matrices $\MD_d(\C)$ is (weakly) closed, $\theta = \theta^* \in \MD_d(\C)$ and 
$\lim_{N \to \iy} \mu_N = \theta$. The fact that $\theta$ is the \emph{unique} invariant state of $\E[\Phi]$ follows again from Theorem \ref{thm:bjm}.
\end{proof}
\begin{rem}
When comparing the preceding theorem with the Proposition \ref{prop:conv_1_unique}, one notes that the hypotheses are relaxed here, asking only that the eigenvalue 1 is simple, without further constraints on the peripheral spectrum. This is due to the fact that we are considering C\'esaro means and fluctuations (such as the ones seen in Example \ref{eg:periph}) cancel out at the limit.
\end{rem}

We now move on to apply the preceding general result to the setting described in the beginning of this section. Recall that the successive interactions were described by i.i.d. random quantum channels $\Phi_n = \Phi^{\beta_n}$, where 
\[\Phi^\beta(\rho) = \trace_\K [U(\rho \otimes \beta)U^*].\]
Since the previous equation is linear in $\beta$, $\E[\Phi^\beta] = \Phi^{\E[\beta]}$ and one gets the following corollary.
\begin{cor}
Let $\{\beta_n\}_n$ be a sequence of i.i.d. random density matrices and consider the repeated quantum interaction scheme with constant interaction unitary $U$. Assume that, with non-zero probability, the induced quantum channel $\Phi^{\beta}$ has an unique invariant state. Then, $\P$-almost surely, for all initial states $\rho_0 \in \MD_d(\C)$, one has
\[\lim_{N \to \iy} \frac 1 N \sum_{n=1}^N [\Phi^{\beta_n} \circ \cdots \circ \Phi^{\beta_1}](\rho_0) = \theta,\]
where $\theta \in \MD_d(\C)$ is the unique invariant state of the deterministic channel $\Phi^{\E[\beta]}$. In particular, if $\E[\beta] = \I_{d'}/{d'}$, then $\theta$ is the ``chaotic'' state ${\I_d}/{d}$.
\end{cor}

\section{Repeated interactions with i.i.d. unitaries}\label{sec:iid_unitaries}

We now consider a rather different framework from the one studied in Sections \ref{sec:fixed} and \ref{sec:random_env}. We shall assume that the interaction unitaries $U_n$ acting on the coupled system $\H \otimes \K$ are random independent and identically distributed (i.i.d.) according to the unique invariant (Haar) probability measure $\Haar_{dd'}$ on the group $\U (dd')$. This is a rather non-conventional model from a physical point of view, but it permits to relax hypothesis on the successive states of the environment and to obtain an ergodic-type convergence result.

As before, we start with a fixed state $\rho_0 \in \MD_d(\C)$. The $n$-th interaction is given by $\rho_{n} = \Phi^{U_n, \beta_n}(\rho_{n-1})$, where $(\beta_n)_n$ is a (possibly random) sequence of density matrices on $\K$ and $(U_n)_n$ is a sequence of i.i.d. Haar unitaries of $\U(dd')$ independent of the sequence $(\beta_n)_n$. Note that we make no assumption on the joint distribution of the sequence $(\beta_n)_n$; in particular, the environment states can be correlated or they can have non-identical probability distributions. The state of the system after $n$ interactions is given by the forward iteration of the applications  $\Phi^{U_n, \beta_n}$:
\begin{equation}\label{eq:cocycle}
\rho_n = \Phi^{U_n, \beta_n} \circ \Phi^{U_{n-1}, \beta_{n-1}} \circ \cdots \circ \Phi^{U_1, \beta_1} \rho_0.
\end{equation}
Since we made no assumption on the successive states of the environment $\beta_n \in \MD_{d'}(\C)$, the sequence $(\rho_n)_n$ is not a Markov chain in general. Indeed, the density matrices $(\beta_n)_n$ were not supposed independent, hence $\beta_{n+1}$ (and thus $\rho_{n+1}$) may depend not only on the present randomness, but also on past randomness, such as $\beta_{n-1}$, $\beta_{n-2}$, etc. Although the sequence $(\rho_n)_n$ lacks markovianity, it has the following important invariance property. 

\begin{lem}\label{lem:U-iid-same-law}
Let $(V_n)_n$ be a sequence of i.i.d. Haar unitaries independent of the family $\{U_n, \beta_n\}_n$ and consider the sequence of successive states $(\rho_n)_n$ defined in Eq.~(\ref{eq:cocycle}). Then the sequences $(\rho_n)_n$ and $(V_n\rho_n V_n^*)_n$ have the same distribution.
\end{lem}
\begin{proof}
Consider a i.i.d. sequence $(V_n)_n$ of $\Haar_d$-distributed unitaries independent from the $U_n$'s and the $\beta_n$'s appearing in Eq.~(\ref{eq:cocycle}). To simplify notation, we put $\tilde \rho_n = V_n\rho_n V_n^*$. We also introduce the following sequence of (random) $dd' \times dd'$ unitary matrices:
\begin{align*}
\tilde U_1 &= (V_1 \otimes \I) U_1, \\
\tilde U_n &= (V_n \otimes \I) U_n (V_{n-1}^{*} \otimes \I), \quad \forall n \geq 2.
\end{align*}
A simple calculation shows that 
\[\tilde \rho_n = \Phi^{\tilde U_n, \beta_n} \circ \Phi^{\tilde U_{n-1}, \beta_{n-1}} \circ \cdots \circ \Phi^{\tilde U_1, \beta_1} \rho_0.\]
It follows that, in order to conclude, it suffices to show that the family $(\tilde U_n)_n$ is i.i.d. and $\Haar_{dd'}$-distributed (it is obviously independent of the $\beta$'s). We start by proving that, at fixed $n$, $\tilde U_n$ is $\Haar_{dd'}$-distributed. Since the families $(U_n)_n$ and $(V_n)_n$ are independent, one can consider realizations of these random variables on different probability space $U_n:\Omega_n^{1} \to \U(dd')$ and $V_n:\Omega_n^2 \to \U(d)$. For a positive measurable function $f:\U(dd') \to \R_+$, one has (we put $V_0=\I$)
\begin{align*}
\E[f(\tilde U_n)] &= \E[f((V_n \otimes \I) U_n (V_{n-1}^{*} \otimes \I))] =\\
&= \int f((V_n(\omega^2_{n}) \otimes \I) U_n(\omega^1_{n}) (V_{n-1}^{*}(\omega^2_{n-1}) \otimes \I)) d\P(\omega^2_{n})d\P(\omega^1_{n})d\P(\omega^2_{n-1})\\
&= \int  \left( \int f((V_n(\omega^2_{n}) \otimes \I) U_n(\omega^1_{n}) (V_{n-1}^{*}(\omega^2_{n-1}) \otimes \I)) d\P(\omega^1_{n})\right) d\P(\omega^2_{n})d\P(\omega^2_{n-1})\\
&\stackrel{(*)}{=} \int \E[f(U_n)] d\P(\omega^2_{n})d\P(\omega^2_{n-1}) = \E[f( U_n)],
\end{align*}
where we used in $(*)$ the fact that the Haar probability on $\U(dd')$ is invariant by left and right multiplication with constant unitaries. We now claim that the r.v. $\tilde U_n$ are independent. For some positive measurable functions $f_1, \ldots, f_n : \U(dd') \to \R_+$, one has
\begin{align*}
\E\left[\prod_{k=1}^n f_k(\tilde U_k)\right] &= \E\left[\prod_{k=1}^n f_k((V_k \otimes \I) U_k (V_{k-1}^{*} \otimes \I))\right] =\\
&= \int \prod_{k=1}^n f_k((V_k(\omega^2_{k}) \otimes \I) U_k(\omega^1_{k}) (V_{k-1}^{*}(\omega^2_{k-1}) \otimes \I)) \prod_{k=1}^n d\P(\omega^1_{k})d\P(\omega^2_{k})\\
&= \int  \prod_{k=1}^n \left( \int f_k((V_k(\omega^2_{k}) \otimes \I) U_k(\omega^1_{k}) (V_{k-1}^{*}(\omega^2_{k-1}) \otimes \I)) d\P(\omega^1_{k}) \right) \prod_{k=1}^n d\P(\omega^2_{k}) \\
&\stackrel{(**)}{=} \int \E[f_k(U_k)] \prod_{k=1}^n d\P(\omega^2_{k}) =  \prod_{k=1}^n \E[f_k(U_k)] \stackrel{(***)}{=}  \prod_{k=1}^n \E[f_k(\tilde U_k)].
\end{align*}
Again, we used in the equality $(**)$ the invariance of the $dd'$-dimensional Haar measure and in $(***)$ the fact that $U_k$ and $\tilde U_k$ have the same distribution.

\end{proof}

We conclude from the above result that although the successive states of the small system $(\rho_n)_n$ are random density matrices that can be correlated in a very general way, their joint probability distribution is invariant by independent unitary basis changes. In other words, the correlations manifest only at the level of the spectrum, the matrices being independently rotated by random Haar unitaries. The ergodic convergence result in such a case is established in the following proposition. 

\begin{prop}
Let $(\tau_n)_n$ be a sequence of random density matrices (we make no assumption whatsoever on their distribution) and $(V_n)_n$ a sequence of i.i.d. Haar unitaries independent of $(\tau_n)_n$. Then, almost surely,
\[\sigma_n = \frac{V_1 \tau_1 V_1^* + \ldots + V_n \tau_n V_n^*}{n} \underset{n \to \iy}{\longrightarrow}  \frac{\I_d}{d}.\]
\end{prop}
\begin{proof}
Since both sides of the previous equation are self-adjoint matrices, it suffices to show that for any self-adjoint operator $A \in \M_d(\C)$ we have $\lim_{n \to \iy} \trace[\sigma_n A] = \trace[A]/d$. Using the invariance of the Haar measure, one can assume that the observable $A$ is diagonal $A = \sum_{i=1}^d s_i \ketbra{e_i}{e_i}$ in some fixed orthonormal basis $\{e_i\}_{i=1}^d$ of $\C^d$. In the same basis, we write  $\tau_k = (t_{i,j}^{(k)})_{i,j=1}^d$ and $V_k = (v_{i,j}^{(k)})_{i,j=1}^d$. To simplify notation, we put
\[\trace[\sigma_n A] = \frac{T_1 + \cdots + T_n}{n},\]
where $T_k = \trace[V_k \rho_k V_k^* A] = \sum_{i_1,i_2,j=1}^d t^{(k)}_{i_1,i_2} s_j v_{i_1,j}^{(k)} \overline{v_{i_2,j}^{(k)}}$. Using the fact that 
\[\E\left[v_{i,j}^{(k)}\overline{v_{i',j'}^{(k)}}\right] = \delta_{i,i'} \delta_{j,j'} \frac 1 d,\] 
one can easily check that that the random variables $T_k$ have mean $\trace[A]/d$, finite variance (a rough bound for $\E[T_k^2]$ is $\trace[A]^2$) and that they are not correlated ($\cov(T_k, T_{k'})=0$, if $k\neq k'$). It is a classical result in probability theory that in this case the (strong) Law of Large Numbers holds and thus, almost surely,
\[\lim_{n \to \iy} \trace[\sigma_n A] = \frac{\trace[A]}{d}.\]
\end{proof}
Putting the previous proposition and Lemma \ref{lem:U-iid-same-law} together, one obtains the main result of this section, an ergodic-mean convergence result for the sequence of states of the ``small'' system.
\begin{prop}
Let $(\rho_n)_n$ be the successive states of a repeated quantum interaction scheme with i.i.d. random unitary interactions. Then, almost surely, 
\[\lim_{n \to \iy }\frac{\rho_1 + \ldots + \rho_n}{n} = \frac{\I_d}{d}.\]
\end{prop}

\bigskip
\textbf{Acknowledgments}
Both authors would like to thank their advisor Professor St\'ephane Attal for his guidance, constant advice and support throughout their graduate careers.

\end{document}